\documentclass[11pt,reqno]{amsart}

\usepackage[utf8]{inputenc}

\usepackage{amsmath,amssymb,amsfonts,xcolor,color} 
\usepackage[colorlinks = true,
            linkcolor = blue,
            urlcolor  = blue,
            citecolor = red,
            anchorcolor = blue]{hyperref}
\usepackage{qtree}
\usepackage{proof}

\newcommand{\F}{\mathcal{F}}

\newcommand{\dec}[2]{dec_{#1}(#2)}

\newcommand{\SL}{\mathcal{SL}}

\renewcommand{\a}{\alpha}
\renewcommand{\b}{\beta}
\renewcommand{\c}{\gamma}

\newcommand{\f}{\varphi}
\newcommand{\p}{\psi}
\newcommand{\x}{\chi}

\newcommand{\lang}{\L}
\renewcommand{\L}{\mathcal{L}}

\newcommand{\spc}{\Pi}
\newcommand{\at}{\mathit{At}}

\newcommand{\fm}{\SL}

\newcommand{\init}{\mathit{Supp}}
\renewcommand{\sf}{\mathit{S}}

\newcommand{\zinf}{*}
\newcommand{\contr}{\bot}

\usepackage{enumerate}
\newtheorem{definition}{Definition}
\newtheorem{proposition}{Proposition}
\newtheorem{theorem}{Theorem}
\newtheorem{lemma}{Lemma}
\newtheorem{example}{Example}
\newtheorem{remark}{Remark}

\newcommand{\paolo}[1]{#1}
\newcommand{\hykel}[1]{#1}

\title{A Logic-based Tractable Approximation of Probability}
\author{Paolo Baldi \and Hykel Hosni}

\address{Department of Philosophy, University of Milan}
\email{paolo.baldi,hykel.hosni@unimi.it}
\begin{document}

\maketitle

\begin{abstract}
We provide a logical framework in which a resource-bounded agent can be seen to perform approximations of probabilistic reasoning. Our main results read as follows. First we identify the conditions under which propositional probability functions can  be approximated by a hierarchy of depth-bounded Belief functions. Second we show that under rather palatable restrictions, our approximations of probability lead to uncertain reasoning which, under the usual assumptions in the field, qualifies as tractable. 
\end{abstract}

\section{Introduction and motivation}
Paraphrasing Laplace's famous dictum, probability depends both on the information we do possess and on the information we do not. Underlying this picture are two ideas:  first, probability functions quantify uncertainty with the operational meaning of ``rational'' degrees of belief \cite{DeFinetti1931,Savage1972,Lindley1982,Paris1994} and second, information \emph{resolves} fully uncertainty. In other words, probability provides a principled estimate of an agent's uncertainty about events whose truth-conditions are fully decided by information as yet not possessed by the agent.

For a straightforward illustration which is also useful to fix some basic terminology, take an urn with $n$ balls which only differ in colour, say yellow, red and green. \paolo{Assume that an agent is further equipped with the information that 1/3 of the balls are yellow, 1/3 are red, 1/3 green, and is asked to draw a ball from the urn. The agent is uncertain about the colour of the particular ball that will be drawn from the urn, but knows which conditions will need to obtain for, say the event ``the ball is red'' to have happened. Furthermore, she knows the distribution of the balls in the urn, and will thus readily assess the likelihood of each particular outcome of this kind of experiment, i.e. will rationally assess her degree of belief on the event e.g. ``the ball is red'', $Bel(Red)$ for short, equal to the proportion of red balls on the urn, that is $Bel(Red) = 1/3$. The same will clearly apply for $Bel(Green)$ and $Bel(Yellow)$. }

\paolo{ Now, let us imagine a second scenario, where an agent only knows that the balls in the urn are either yellow, red or green, without knowing their distribution. Here, the ``Laplacian'' way of looking at the problem will translate the lack of information about the composition of the urn into very precise probabilistic (as it turns out) information about the event of interest: since no further information on the composition is indeed available, then there is no reason to depart from the assumption that colours are uniformly distributed in the urn, which readily gives $Bel(Red)=1/3$, and similarly for \emph{Green} and \emph{Yellow}.}

\paolo{Both scenarios sketched so far lead to the same assignment, but relying on different kinds of information. This is however problematic, as shown by the following example, where both types of situations coexist in the same context. 
Take another three-colours urn composed for 1/3 of red balls and for 2/3 of balls which are either yellow or green, in unknown proportion.
 Whilst the assignment of belief to $Red$ will follow the pattern of the first scenario, and, using the provided information about the proportion of red balls will lead to  $Bel(Red)=1/3$, the assignment of belief to the other two colours would follow the pattern of the second scenario. 
A reasoning based on the lack of information will lead indeed to
$Bel(Yellow)=Bel(Green)=1/3$, and again, we will have $Bel(Yellow)=Bel(Green)= Bel(Red) =1/3$, with an assignment motivated by two different kinds of information (and lack thereof). }

 Over the past century many have taken issue with the Laplacian line of argument above. Among them \cite{Keynes1921,Shafer1976a,Walley1991} had a profound and continuing impact on the foundations of uncertain reasoning, whereas \cite{Knight1921,Ellsberg1961,Schmeidler1989,Jaffray1989} brought related problems to the attention of decision theory. More recently, similar concerns have been raised in artificial intelligence from a variety of distinct angles, among which \cite{Dubois2001c,Kyburg2001} had significant impact. With regard to this growing literature, readers may find out more in the general exposition provided by \cite{Halpern2003} and an up-to-date overview \cite{Marquis2020}.

Whilst the formal models which originated from the above criticisms may diverge significantly, a unifying feature which is easily found in all of them is the criticism of the additivity forced upon degrees of belief by the probabilistic representation of uncertainty. This, as will be developed in the remainder of this section, can be undesirable in those situations in which telling precisely the information an agent possesses from that which it doesn't, is neither practically nor conceptually straightforward. The research reported in this paper aims at adding a logical dimension to this long-standing issue in the foundations of uncertain reasoning.


Unlike the above mentioned criticisms, we do not take issue directly with the additivity of the measure representing uncertainty. Rather, we focus on an important but very much neglected element in the analysis of the quantification/resolution view of uncertain reasoning, namely the fact that it requires certain \emph{inferential abilities}.
 \paolo{These, under suitable conditions, can give rise to \emph{information}, which eventually leads the agent's degree of belief converge to an additive measure -- i.e. a probability.} By tackling the problem at its logical roots we also hope to encourage logicians to join this very interesting multidisciplinary conversation of the very foundations of rational reasoning and decision-making.

To motivate the need for a logical approach to the problem,  it is useful to point out explicitly the way \emph{classical} logic constraints (finitely additive) probability functions. So take
$\vdash$ to be the classical consequence relation on a (finite) propositional language $\lang$. A \emph{probability function on $\lang$}  is a normalised and finitely additive function
from the sentences of the language to the real unit interval, i.e. one satisfying the following conditions:
\begin{enumerate}
\item[(PL1)] If $\vdash \theta$  then $P(\theta)=1$,
\item[(PL2)] If $\vdash \neg (\theta\wedge \f)$ then  $P(\theta\vee
  \f)=P(\theta)+ P(\f)$.
\end{enumerate}
The Stone representation theorem for Boolean algebras and the
representation theorem for probability functions (see e.g. \cite{Paris1994})
guarantee that the Kolmogorov measure-theoretic axiomatisation of probability is
equivalent to the logical one above. Framing the axioms logically however puts forward the role  of $\vdash$ in giving probability functions an unwelcome inability to distinguish between ``probabilistic knowledge'' and  ``probabilistic ignorance''. To see this, recall that $\vdash \theta\vee\neg\theta$, and that (PL1) and (PL2) immediately give us 
\begin{equation}
P(\neg \theta)=1-P(\theta).\label{eq:lem}
\end{equation}
This amounts to saying that in the framework we have no way of representing a very common situation: the agent doesn't known anything about $\theta$. 
Instead of allowing the agent to represent its lack of opinions about the event $\theta$, \eqref{eq:lem}
effectively translates absence of information about $\theta$ into information about its negation.
Whilst this is easily seen to be inappropriate in many realistic scenarios, to the best of our knowledge \cite{DAgostino2017} was the first paper to point out that the root of the problem, and hence a robust way to fix it, was not to be found (exclusively) in the axioms of probability functions, but at the more fundamental level of the underlying consequence relation.

This is apparent by focusing on another immediate consequence of (PL1), (PL2):
\begin{equation}
  \theta \vdash \f \text{ implies }  P(\theta)\le P(\f),\label{eq:mon}
  \end{equation}
  which expresses the monotonicity of probability functions with respect to $\vdash$. This property is mathematically convenient,  but puts on the agent a very heavy burden which owes to its intractability. Indeed, a logical deduction from $\theta$ to $\f$ might be highly nontrivial and hard to find, making the application of \eqref{eq:mon} a constraint of rationality that realistic agents may not be in a position to comply with. Many authors, including \cite{Savage1967,Hacking1967} have pointed out that imposing that ``rational degrees of belief'' should be deductively closed is too tall an order, but they did not tackle the problem from the logical point of view. \paolo{This is what this paper attempts to do, by drawing on the theory of Depth-bounded Boolean logic (DBBL) investigated in  \cite{DAgostino2013,DAgostino2015}.}

 A characteristic feature of DBBL is its informational nature. Whereas connectives in classical logic are defined in terms of truth-values, DBBL provide an informational view of logical consequence which, as a by-product, also provides a tractable approximation of it. Thus DBBL provide the ideal setting to tackle the shortcomings of the probabilistic quantification of uncertainty highlighted by \eqref{eq:lem} and \eqref{eq:mon} above.

  The central idea behind the (semantic) approach to DBBL is to distinguish two kinds of information. The first is information which the agent possesses explicitly or can trivially infer from it. For  definiteness, it is the kind of information that an agent holds when she holds the information that a conjunction is true, i.e. that both conjuncts are true.
  Put otherwise,  the information that both conjuncts are true is immediately available to the agent possessing the information that the conjunction is true. The second kind of information is more subtle and is \emph{obtained} by making a specific kind of hypothetical reasoning. This is closely related to classical reasoning by cases, with the additional requirement that its disjunctive premises always partition the set of (classical) valuations.

To illustrate this kind of reasoning, let us consider a puzzle, originally due to Hector Levesque, and which has been largely discussed in the psychology of reasoning literature, see e.g.  Chapter 6 of \cite{Stanovich2009}.

\begin{example}\label{ex:ann}
  
Experimental subjects are faced with the following test:
  
\begin{quote}
  Jack is looking at Anne, and Anne is looking at George. Jack is married, but George is not. Is a married person looking at an unmarried person?

  \begin{itemize}
  \item[A)] Yes
  \item[B)] No
  \item[C)] Cannot be determined.
      \end{itemize}
\end{quote}

It is reported by \cite{Stanovich2009} that over 80\% of subjects give a wrong answer,  the vast majority of which being C).

Here is a plausible explanation for this systematic mistake, which was certainly engineered in the problem. The subject is provided with some explicit information, which in first-order logic notation can be represented as follows:
\begin{eqnarray}
  \label{eq:1}
  Looks(Jack,Anne)\\
  \label{eq:2}
  Looks(Anne,George)\\
  \label{eq:3}
 Married(Jack)\\
  \label{eq:4}
\neg Married(George).
\end{eqnarray}
Hence there is no information which is provided about the marital status of Ann, which probably leads most respondents to this frequent mistake.

The correct answer is instead A), which can be arrived at by performing a particular kind of hypothetical reasoning based on the fact that
\begin{equation}\label{eq:biv}\tag{Biv}
\text{either } Married(Anne) \text{ or } \neg Married (Anne).
\end{equation}
Those mutually exclusive disjuncts clearly exhaust the set of classical models. So we get two cases, both of which decide positively the question of interest, namely
\paolo{$$\exists x\exists y \f(x,y) \qquad \text{ where } \f(x,y) := Looks(x,y)\land Married(x) \land \neg Married(y) .$$}

For we have, from \eqref{eq:biv} and the information in \eqref{eq:1}-\eqref{eq:4} that either
$$
\gamma:= Married (Anne) \wedge \neg Married(George) \wedge   Looks(Anne,George)$$
or
$$\gamma':= \neg Married (Anne) \wedge  Married(Jack) \wedge   Looks(Jack,Anne),
$$
holds, with $\f(Anne,George)$ following immediately from $\gamma$, and $\f(Jack,Anne)$ following immediately from $\gamma'$.

\end{example}
This example singles out two senses of ``information''. The first is conveyed by \eqref{eq:1}-\eqref{eq:4} which is readily seen as information actually held by the experimental subjects. In addition, subjects can be seen to attain logical conclusions which exceed the information actually held, \paolo{by means of hypothetical reasoning}. This essentially amounts to applying a single instance of the reasoning by cases on a pair of jointly exhaustive alternatives - as pointed out in \eqref{eq:biv}, either Anne is married or she is not.





The hierarchy of Depth-bounded Boolean logic arises by bounding the number $k$ of allowed nested iterations of hypothetical inferential steps of the kind just illustrated. If $k=0$, then only information actually held by the agent can be used as the premises of a logical deduction, yielding the $0$-depth consequence relation denoted by $\vdash_0$.  So, an agent confined to a $0$-depth consequence relation, would be unable to solve the problem of Example \ref{ex:ann}, which is readily solved above at depth 1.

For $0\le k < m$ it can be shown that $\vdash_k\ \subset \ \vdash_m$, whereas $\lim_{k\to\infty}\vdash_k \ = \ \vdash$. This justifies interpreting the hierarchy of DBBLs as an approximation to classical logic, which is indeed attained when the agent is allowed an unbounded use of hypothetical information along the lines illustrated. 
As we shall point out in due course, keeping this bound in place also results in each  element $\vdash_k$ of the sequence being tractable, in the sense of computational complexity. 

\paolo{ While various approaches in the literature  \cite{Finger06,Levesque2020} have similarly focused on feasible approximation of logical systems, what sets apart DBBLs from those is its clear information-based semantics: this provides indeed an uncertainty resolution mechanism which is naturally more adequate than a truth-based semantics, for tackling our motivating questions.}
Recall that an \emph{event}, understood classically, is a binary random variable. This is captured by the (classical) valuation
\begin{equation*}
 v(\theta)=
 \begin{cases}
   1 \text{ if the event obtained}\\

   0 \text{ otherwise},
 \end{cases}
\end{equation*} 
which represents the event's indicator function. In this classical picture, uncertainty is resolved by a logical valuation, which is a total function from the language to the binary set $\{0,1\}$. In probabilistic modelling, this amounts to the usual assumption concerning the ``deterministic nature'' of the phenomena whose uncertainty is being quantified. See e.g. \cite{DeFinetti1974} for an authoritative point of view along these lines.

By contrast, an information-based uncertainty resolution accounts explicitly for the possibility that the agent is not informed about a particular event. This can be illustrated by decorating the valuation symbol used above with a superscript $i$, meant to suggest ``information'':
\begin{equation*}
 v^i(\theta)=
 \begin{cases}
   1 \text{ if the agent is informed that $\theta$ obtained}\\
   0 \text{ if the agent is informed that $\neg\theta$ obtained}\\
  \zinf \text{ if the agent is not informed about $\theta$ }.
 \end{cases}
\end{equation*}

So in Example \ref{ex:ann} above, it is natural to represent the fact that the marital status of Ann is unspecified by $v^i(Married(Anne))=\zinf$. However, a single binary branching leads to the hypothetical resolution of this uncertainty, with one branch deciding $v^i(Married(Anne))$ positively, and the other negatively.

\paolo{The key difference between $v$ and an evaluation such as $v^i$ lies in the fact that the latter allows for a bounded number $k$ of iterations of hypothetical uncertainty resolution, and, in particular when $k=0$ may  provide only a \emph{partial} evaluation of the sentences/events.} This  in turn leads to a ``non-deterministic'' semantics \cite{DAgostino2013} which is instrumental to avoiding the automatic translation of the agent's ignorance about $\theta$ to information about $\neg\theta$.  


In analogy to  DBBLs, this paper investigates a hierarchy of \emph{Depth-bounded Belief functions} which approximate probability functions and asymptotically coincide with them. Our construction is inspired by the theory of Dempster-Shafer Belief functions \cite{Shafer1976a}, as suggested by our choice of terminology. As in DS-theory,  none of the belief functions in our hierarchy is constrained by additivity  (PL2 above), except for the one attained in the limit. More precisely, given two
  incompatible events $\f$ and $\p$, a depth-bounded belief function might give $B_k(\f\lor\p)$ a value in $[0,1]$ which is distinct from $B_k(\f)+B_k(\p)$. Going back to our comments on \eqref{eq:lem} above, the non-additive nature of our belief functions is necessary to ensure a distinct representation of `knowledge' and `ignorance'. However, unlike DS-theory (and related proposals) it is achieved {primarily} by replacing classical logic as the uncertainty resolution mechanism -- i.e. by changing  `the logic of events'--  rather than by {just} relaxing the axiomatic properties of the uncertainty measure leaving the underlying logic classical, as put forward in the position paper \cite{DAgostino2017}.

Finally, a very interesting feature of the framework to be introduced below is in our view this: endowing the agent with higher logical abilities (as captured by the index of the relation $\vdash_k$) allows her to obtain increasingly tighter approximations of $B_k(\f)+B_k(\p)$ by $B_k(\f\lor\p)$. This puts forward a seemingly novel approach, a logic-based one, to Ellsberg-like problems, as illustrated by Example \ref{ellsberg} below.

\paolo{ As a welcome byproduct, major reasoning tasks which are based on any member $B_k$ of our hierarchy will turn out to be computationally tractable, in a sense to be clarified in Section \ref{sec:complex}. }


\subsection{Plan of the paper and summary of contributions}
Section \ref{sec:bounded} presents the key features of Depth-bounded boolean logics in a form which suits our present needs. Since this section is instrumental to making this paper self-contained, we refer the interested reader to \cite{DAgostino2013,DAgostino2015} for the full logical details and \cite{DAgostino09} for the philosophical underpinnings of the theory of DBBL. Section \ref{sec:approximating-probability-functions}  introduces the belief functions which provide, through our main results, an approximation to probability functions. 
\paolo{The resulting approximation will be immune, by logical design, to the problem of turning the absence of information about an event into information about its negation.}
Section \ref{sec:complex} investigates the complexity of our approximation and
identifies the reasoning tasks with depth-bounded belief functions which are feasible (under the usual assumptions which will be briefly recalled). Section \ref{sec:conclusion} collects some final remarks and points to the directions for future work on this topic. 

Our main results read as follows. Theorem \ref{thm:approxprob} shows that each probability function can can be approximated by a hierarchy of Depth-bounded Belief functions, while, conversely, Theorem \ref{thm:prob} singles out the conditions under which our Depth-bounded Belief functions actually determine a probability in the limit. Theorems \ref{poly0},\ref{polyk}, and \ref{inf} show that under rather palatable restrictions, the depth-bounded \paolo{functions} introduced here are an adequate tool to tackle the well-known unfeasibility of logic-based uncertain reasoning \cite{Paris1994,sipser}.

The framework and results presented below extend significantly those of  \cite{Baldi2020a}, which is therefore superseded by the present paper. In the related  \cite{Baldi2020} we put forward a logic-based framework which served as a direct inspiration for the one introduced below, and whose aim was the approximation of DS {Belief functions}. So the main result of \cite{Baldi2020} is that the hierarchy of approximating measures asymptotically coincides with DS Belief functions.
Note that elements in the hierarchy of depth-bounded belief functions in \cite{Baldi2020} only provide a logical model of agents with increasingly higher inferential ability, which nonetheless has no impact on the resolution of uncertainty. This is reflected formally by the fact that, even in the limit, the Depth-bounded belief functions in \cite{Baldi2020} need not be additive. This work, in contrast, puts the hierarchy of depth-bounded belief functions in correspondence with an increasingly higher ability to resolve uncertainty \paolo{beyond}
 the information actually held by the agent. This leads in the limit to \emph{additive} measures. Furthermore, \cite{Baldi2020} provides no analysis of the computational complexity of Depth-bounded belief functions (though it is conjectured in there that they are actually intractable, in the most general case).


\section{Bounded sequences and Bounded logics}\label{sec:bounded}
We begin with some conventions and terminology. We will denote by $\L$ a 
finite propositional language, with variables $p,q$ etc., connectives $\neg, \wedge,\vee$, and the constant $\contr$, standing for any contradiction. The set of
sentences generated recursively from $\L$ will be denoted by $\fm$.   Lowercase Greek letters denote sentences, and uppercase Greek letters sets of
sentences.  
Given any sentence $\f\land\p,\f\lor\p$ in $\fm$, the sentences $\f,\p$ are said to be 
\emph{the immediate subsentences} of $\f\land \p, \f\lor \p$, respectively. Similarly, for any sentence $\neg\f$, its immediate subsentence is $\f$. A propositional variable has no immediate subsentences. The set of
\emph{subsentences} of $\f$ is denoted by $S(\f)$ and is the smallest set containing $\f$ and closed under
immediate subsentences -- ditto for $S(\Gamma)$, where $\Gamma$ is a set of sentences. Finally $At_\lang$ denotes the set of \emph{atoms}, i.e. the maximal (classically) consistent conjunctions of literals from $\lang$.

Trees, and sequences thereof, will represent \paolo{the information}
 intervening in depth-bounded reasoning. In our framework the root of a tree is labelled with the information initially held by the agent. If no such information is available to the agent, the root will be labelled by $\zinf$. Otherwise it will be labelled by a set of sentences from $\mathcal{SL}$. 
\paolo{In what follows, the leaves of a tree $T$ are denoted by $Le(T)$, and the \emph{depth of a node} $\a$ in a tree is the length of the path from the root of the tree to $\a$.}
 
Branching is where the core construction \paolo{of the other kind of information, the hypothetical one}, leading to uncertainty resolution, takes place, along the line of our first definition:


\begin{definition}\label{Tk}
Let  $\gamma \in \fm \cup\{\zinf\}$. 
A \emph{depth-bounded tree sequence ${(T_k)}_{k\in\mathbb{N}}$ (DBT-sequence for short) rooted in $\gamma$} is defined recursively as follows:
\begin{itemize}
\item $T_0$ is a tree having as only node its root $\gamma$.
\item for $k> 0$, $T_{k}$ is obtained by branching at least one node $\a$ 
of depth $k-1$  in $T_{k-1}$ 
with two nodes $\a\land \b$ and $\a\land \neg \b$, \paolo{for some $\b\in \fm$.}
\end{itemize}
A \emph{ finite depth-bounded tree sequence} is any initial segment ${(T_k)}_{k<n}$ of a depth-bounded sequence ${(T_k)}_{k\in\mathbb{N}}$, for some $n\in \mathbb{N}$.
\end{definition}
\paolo{A few consideration on the definition above are in place. First, the formula $\b\in\fm$, used for the branching in the Definition above, is not fixed, but different formulas can be used for each different parent node and each iteration of the construction of the $T_k$s. Second, in case $\gamma$ is the empty information $\zinf$, we will take $T_1$ to be a tree rooted in $\zinf$, with two branchings $\b$ and $\neg \b$, for some sentence $\b\in\fm$. Finally, note that at least one formula $\a\in Le(T_k)$ has to be of depth $k$, but this need not be the case for all the formulas in $Le(T_k)$.
Consider indeed the following.}

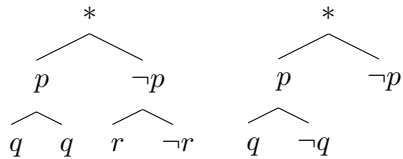
\begin{example}\label{ex:depth}
\[
\Tree[.$\zinf$ [.\ $p$ [.\ $q$ ] [.\ $q$ ]] [.\ $\neg p$ [.\ $r$ ] [.\ $\neg r$ ]]]    \ \  \ \
\Tree[.$\zinf$ [.\ $p$ [.\ $q$ ] [.\ $\neg q$ ]] [.\ $\neg p$ ]]
\]
Both the tree above on the left, that we will denote by $T_2$, and that on the right, that we will denote by $S_2$, are trees of depth 2, obtained by branching a tree (say $T_1=S_1$) composed of $\zinf$ and children nodes $p,\neg p$. 
However, $T_2$ is obtained by branching each formula in $Le(T_1)$, so that each node in $Le(T_2)$ has depth $2$. 
On the other hand, $S_2$ is obtained by expanding only the node $p$, of depth 1, in $Le(S_1)$. The node  $\neg p \in Le(S_1)$ is not expanded, so that we will also have $\neg p \in Le(S_2)$, while $\neg p$ is still of depth 1. Note that it will not suffice to expand the latter node to obtain a tree of depth 3, but rather at least one of the nodes of depth 2 in $Le(S_2)$ has to be expanded.
\end{example}
\hykel{The characteristic feature of Definition \ref{Tk} is {that} trees form an ordered \emph{sequence}, which can be interpreted as a hierarchy of levels of hypothetical uncertainty resolution branchings, of the kind illustrated by Example \ref{ex:ann}. Each successive level is obtained through a new binary branching aimed at deciding a yet undecided sentence and its negation. In this sense each new level of hypothetical uncertainty resolution is achieved by means of a constrained version of reasoning by cases. 
Put in yet another way, step $k+1$ in the recursive construction of ${(T_k)}_{k\in\mathbb{N}}$ amounts  to the binary expansion of some \mbox{$k$-depth} leaves by new children nodes, each being labeled by a pair of mutually inconsistent  pieces of information.
Hence depth-bounded sequences can be seen as possible explorations of how uncertainty may be resolved by information which is \emph{not} held by the agent, starting from the information which the agent does hold.
Clearly, agents whose reasoning is constrained by distinct bounds may not resolve the relevant uncertainty in the same way. This provides the logical basis  to add in Section \ref{sec:approximating-probability-functions}  an important and as yet under appreciated dimension to the the Laplacian dictum to the effect that probability depends in part on ``knowledge'' and in part on ``ignorance''.}


Going back to the construction in Definition \ref{Tk}, note that if a leaf contains an atom as a subsentence and is not contradictory, then \emph{all} uncertainty in the sentences of interest has been resolved, i.e. the information reached coincides with a classical truth-value. However, and this is our main interest here, the agent will not normally possess sufficient resources to get there, and hence will be forced to content herself with an approximation of such truth-value. As we prove in the next section, the greater the depth of the tree-sequence, the better the agent will approximate a probabilistic estimate of such an uncertainty.  

\paolo{We must now integrate our framework with a sensible way of applying the recursive step in Definition \ref{Tk}, i.e. for selecting only limited, \emph{useful}} amount of branchings. 
A criterion of ``usefulness'' emerges rather naturally: given a sentence $\f$ in classical logic that an agent might want to decide, the agent checks which leaves $\a$ classically derive $\f$, i.e. $\a\vdash\f$. However, this is clearly not viable for resource-bounded agents. First, owing to the intractability of classical logic. Second, as pointed out in  \cite{DAgostino2015} and informally in the introduction to this paper, a classical derivation $\a\vdash\f$ might itself make ``hidden''  use of hypothetical information, beyond that encoded in $\a$. 
This would be undesirable given our goal of accounting for the distinct roles played by actual and hypothetical information in the construction of the agent's uncertainty quantification. Hence we resort 
to 0-depth inference, introduced in \cite{DAgostino2013}. This logic, as mentioned in the introduction, does allow any use of hypothetical information, and, as a welcome byproduct, it is computationally tractable.

So, in preparation for Proposition \ref{treeandlogic}, which collects the key properties of depth-bounded sequences, we briefly recall now the basic features of 0-depth inference from a  proof-theoretic point of view. The presentation is entirely based on the IntElim (Introduction and Elimination) rules in Table \ref{intel}.  We note in passing that this logic, and indeed the whole hierarchy of depth-bounded boolean logics which stems from it, is sound and complete with respect to a non-deterministic, information-based semantics. See \cite{DAgostino2015,DAgostino2013} for more details.
\begin{table}[h!]
\centering
\begin{tabular}{|ccc|}
\hline
& & \\[1.0em]
$\infer[(\wedge \mathcal{I})
]{\f\land \p}{\f & \p } $  
&  
         $\infer[( \neg \wedge \mathcal{I}1)]{ \neg (\varphi \wedge \psi)}{\neg \varphi}$
&
$\infer[( \neg \wedge \mathcal{I}2)]{\neg (\varphi \wedge \psi)}{\neg \psi}$
\\ & & \\ 
               $\infer[ ( \neg \vee \mathcal{I})]{\neg(\varphi \vee \psi)}{\neg \varphi &    \neg \psi}$
  &     
    $\infer[ (\vee \mathcal{I}1)]{\varphi \vee \psi}{\varphi}$
      &
       $  \infer[ (\vee \mathcal{I}2)]{\varphi \vee \psi}{\psi}$
      \\ & & \\ 
  $\infer[ (\contr \mathcal{I})]{\contr}{ \varphi \ \     \neg \varphi}$
 & 
              $\infer[ (\neg \neg \mathcal{I})]{\neg \neg \varphi}{\varphi}$
       &
\\ & & \\ 
        $\infer[ (\vee \mathcal{E}1)]{\psi}
{\varphi \vee \psi \ \
     \neg \varphi}
$
           &
             $ \infer[  (\vee \mathcal{E}2)]
{\varphi}{\varphi \vee \psi \ \
     \neg \psi}$
           &      
$\infer[ (\neg \vee \mathcal{E}1)]{\neg \varphi}
{\neg (\varphi \vee \psi)}$
  \\  & & \\ 
$      \infer[ (\neg \vee \mathcal{E}2)]{\neg \psi}{\neg (\varphi \vee \psi)}$
 &  
       $  \infer[ (\wedge \mathcal{E}1)]{
     \varphi}{\varphi \wedge \psi
}$
           &
         $\infer[ (\wedge \mathcal{E}2)]{\psi}{\varphi \wedge \psi}$
                \\  & & \\ 
          $ \infer[ (\neg \wedge \mathcal{E}1)]
     {\neg\psi}{\neg (\varphi \wedge \psi) \ \
 \varphi}$
         \  & \ 
           $ \infer[(\neg \wedge \mathcal{E}2)]{
     \neg\varphi}{\neg (\varphi \wedge \psi) \ \
     \psi}$
      &
          \\ & &  \\ 
        $ \infer[ (\neg \neg \mathcal{E})]{\varphi} {\neg \neg \varphi}$
   &
$      \infer[ (\contr \mathcal{E})]
     {\varphi}{\contr}$
  &
\\ & &  \\ 
\hline
\end{tabular}

\caption{The standard propositional Introduction and Elimination rules for $\neg,\wedge,\vee$.}
\label{intel}
\end{table}

\begin{definition}\label{0depthlogic}
For 
$\Gamma\cup \{\f\}\subseteq \fm$, let
 $\Gamma \vdash_0 \f $ if and only if there is a sequence of  sentences $\f_1,\dots,\f_m$ such that $\f_m = \f$ and each sentence $\f_i$ is either in $\Gamma$ or obtained by an application of the rules in Table \ref{intel} on the sentences $\f_j$ with $j<i$. 
 \end{definition}

So whilst $p\wedge q\vdash_0 p$, it is not the case that  $\zinf\vdash_0 p\vee\neg p$.
More generally, it can be shown that, contrary to classical logic,  $\vdash_0$ has no tautologies at all, i.e. for no $\f$, {do we have} $\zinf\vdash_0\f$. 

The remainder of this section illustrates the pivotal role of  0-depth derivation in the construction of  bounded sequences, starting with the definition of \emph{maximal sequences} which pins down the most useful form of uncertainty resolution.

As usual, we say that $\c$ decides $\f$ if either $\c\vdash_0\f$ or $\c\vdash_0 \neg \f$. It is convenient to pin down the elements of a set $\Gamma$ which decide the sentence $\f$,  i.e.  we write \[ \dec{\Gamma}{\f} = \{\a \in\Gamma\mid \a \text{ decides } \f \}.\] This allows us to define maximality and closure in trees and sequences, as follows.

\begin{definition}\label{max}
  Let $\{T_k\}_{k\in\mathbb{N}}$ be a DBT-sequence
and $\f\in \fm, \spc\subseteq \fm $.
\begin{itemize}
\item A tree $T_k$ is said to be
  \emph{$\{\f\}$-maximal} if there exists no tree $T'_k$  
of depth $k$ with the same root, such that $|\dec{Le(T'_k)}{\f}| > |\dec{Le(T_k)}{\f}|$.
\item A tree $T_k$ is said to be
  \emph{$\{\f\}$-closed} if  $\dec{Le(T_k)}{\f} = Le(T_k)$, i.e. all the formulas $\a \in Le(T_k)$ decide $\f$.
\item A sequence $\{T_k\}_{k\in\mathbb{N}}$  is \emph{$\spc$-closed} if, for each $\f\in \spc$ there is a $k\in\mathbb{N}$ such that $T_k$ is  $\{\f\}$-closed.
\item A finite sequence $\{ T_k\}_{k\leq n}$ 
is $\Pi$-maximal(closed) if for each $\f \in \Pi$, there is a $j\leq n$ such that $F_j$ is $\{\f\}$-maximal(closed).
\end{itemize}
\end{definition}

Clearly, closed trees are in particular maximal. While a closed tree represents the ideal case in which  an agent has resolved all the uncertainty relatively to a formula of interest, we can think of maximality as a ``second-best''. For DBT sequences, \paolo{it can be shown} that the two concepts actually coincide in the infinite case (this is the reason why we do not include a notion of maximality for them), while they differ for the finite case. Correspondingly, our approximation results, i.e. Theorem \ref{thm:prob} and Theorem \ref{thm:approxprob} make use of assumptions on closed sequences, while those in Section \ref{sec:complex}, where finite sequences are concerned, are limited to maximal ones. 

\begin{example}\label{ex:maxnomax}
Consider the formula  $(p\lor \neg p) \lor q$, and the following depth $1$ trees:

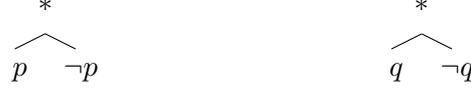
\Tree[.$\zinf$ [.\ $p$ ] [.\ $\neg p$ ]]    \ \  \ \
\Tree[.$\zinf$ [.\ $q$ ] [.\ $\neg q$ ]]\\

The tree on the left-hand side is $\{(p\lor \neg p) \lor q\}$- closed, for we have both $p \vdash_0(p\lor\neg p) \lor q$ and  $\neg p\vdash_0 (p\lor\neg p) \lor q$. However, it is not the case that  $\neg q$ on the tree on the right-hand-side decides $(p\lor\neg p) \lor q$. Hence the tree on the right-hand-side is not $\{(p\lor \neg p) \lor q\}$-closed, and in particular not maximal.

Note that closeness (let alone maximality) does not imply uniqueness. This is easily shown by taking again the $1$ depth trees above, which are both  $\{(p\lor \neg p) \lor (q\lor\neg q)\}$-closed.

%

\end{example}

So far, we have only used 0-depth consequence, which is only the starting point of the hierarchy of Depth-Bounded Boolean Logics. \paolo{The remaining logics in the hierarchy were defined in \cite{DAgostino2013} as follows.}
\begin{definition}\label{kdepthlogic}
For each $k>0$ and set of sentences $\Gamma\cup \{\f\}\subseteq \fm $, we let
$\Gamma \vdash_k \f$ if and only if there is a $\b\in \sf(\Gamma\cup \{\f\})$ such that $\Gamma, \b\vdash_{k-1} \f $ and  $\Gamma, \neg\b\vdash_{k-1} \f $. 
\end{definition}
We refer to the original works \cite{DAgostino2013,DAgostino2015} for further details on the key properties of Depth-Bounded Boolean Logics which include the following:
\begin{itemize}
\item  $\vdash_0\subset\vdash_1\subset\cdots\subset \vdash_k\subset\cdots$, so the depth-bounded consequence relations form a hierarchy;
\item the limit for $ k\rightarrow \infty$ of $\vdash_k$ is $\vdash$, i.e. in the limit, the hierarchy of DBBLs coincides with classical logic;
\item for each $k$,  $\vdash_k$ has a polynomial-time decision procedure.
\end{itemize}

  The next Proposition shows that $k$-depth logics and maximal trees are closely connected. Before that, let us recall that, as some classical inferences are available only at higher level of the hierarchy of DBBLs, the same holds for recognizing  contradictions. In principle, it might be the case that a piece of information $\a$ is contradictory ($\a\vdash\bot$), while this might not be recognizable at 0-depth ($\a\not\vdash_0\bot $). Ruling out this case leads us to a even tighter connection between $k$-depth trees and the corresponding $k$-depth logic. In the following we will say that a tree $T_k$ is \emph{free of deep contradictions} iff for each $\a\in Le(T_k)$, having $\a\vdash\bot$ implies that $\a\vdash_0 \bot$. 

\begin{proposition}\label{treeandlogic}

Let  $\f \in \fm\cup\{\zinf\},\p\in\fm$. 
If $\f\vdash_k \p$, then the following hold: 
\begin{enumerate}[(a)] 
\item There exists a tree $T_{k}$ of depth $k$ rooted in $\f$, such that $\a\vdash_0\p$, for each $\a\in Le(T_{k})$.  
\item There exists a tree $T_{k}$ of depth $k$ rooted in any sentence $\c\in\fm\cup\{\zinf\}$, such that, for each $\a\in Le(T_{k})$, if $\a\vdash_0\f$, then $\a\vdash_0\p$.  
\item For any $\{\p\}-$maximal tree $T_k$ of depth $k$ rooted in $\f$ and free of deep contradictions, we have that $\a\vdash_0\p$, for each $\a\in Le(T_k)$. 
\item For any $\{\neg\f \lor \p\}-$maximal tree $T_{k+1}$ of depth $k+1$ rooted in a sentence $\c$  and free of deep contradictions, we have that $\a\vdash_0\p$, for each $\a\in Le(T_{k+1})$ such that $\a\vdash_0\f$. 
\end{enumerate}
\end{proposition}
\begin{proof}
  
(a).\paolo{ Let us construct a tree $T_k$ of depth $k$ as follows. The root is $\f$. Now, since $\f\vdash_k \p$, by Definition \ref{kdepthlogic}, there is a $\b$ such that $\f,\b\vdash_{k-1}\p$  and $\f,\neg\b\vdash_{k-1}\p$. We thus branch the root with two nodes $\f\land\b$, $\f\land \neg\b$.
We repeat the procedure as follows: for any $1\leq j< k$, assuming that, for each node $\a$ at depth $j$, we have $\a \vdash_{k-j} \p$, we branch the node $\a$ with the two formulas $\b$ and $\neg \b$ such that $\a,\b\vdash_{k-j-1}$
and $\a,\neg \b \vdash_{k-j-1} \p$ (such a $\b$ exists, by Definition \ref{kdepthlogic}). }

\paolo{ Proceeding by induction, we then obtain that any node $\a$ at depth $0\leq j\leq k$ in the tree thus constructed is such that $\a\vdash_{k-j} \p$; hence, in particular,  for any leaf $\a\in Le(T_k)$ (i.e. formula at depth k), $\a\vdash_0\p$.}

(b) \paolo{We consider the same tree as in (a), but with root $\c$, instead of $\f$. Let us write any leaf in this tree as $\gamma\land \b$ and any leaf in the tree in (a) as $\f\land\b$. By (a), we have that  $\f\land\b\vdash_0\p$ for any $\b$. Whenever  $\gamma\land\b\vdash_0\f$, since by $(\land\mathcal{E}) $, $\gamma\land\b\vdash_0\b$, we also have $\gamma\land\b\vdash_0\f\land\b$ (by $\land\mathcal{I}$). Hence, by the transitivity of $\vdash_0$, we obtain our claim, namely that $\gamma\land\b\vdash_0\p$, whenever $\gamma\land\b\vdash_0\f$.}

(c) \paolo{The tree built in (a) is such that  $\a \vdash_0 \p$ for each $\a\in Le(T_k)$, hence it is a $\{\p\}$-closed tree. Any other  $\{\p\}$-maximal tree $T'_k$ will thus also be in particular $\{\p\}$-closed.
Consider now a $\{\p\}$-closed tree $T'_k$ which is free of deep-contradictions: each leaf will decide $\p$. Assume that, for some leaf, say $\f\land\b \in Le(T'_k)$, it holds that $\f\land \b \vdash_0\neg\p$.
 This means in particular that $\f\land\b \vdash_k \neg \p$. But then, since $\f\vdash_k\p$, we also have $\f\land\b\vdash_k\p$. Hence, $\f\land\b\vdash_k\bot$, and by the assumption that $T'_k$ is free of deep contradictions (see Definition \ref{max}), we will have $\f\land\b\vdash_0 \bot$. Hence, by rule ($\bot\mathcal{E}$), we have again $\f\land\b\vdash_0\p$. This finally shows that each leaf in a $\{\p\}$-maximal tree derives at 0-depth $\p$.}

(d)\paolo{Assume $\f\vdash_k\p$.  
Note that $\neg \f \vdash_0 \neg \f \lor \p$ by the rule ($\lor \mathcal{I}1$),  hence in particular $\neg \f \vdash_{k} \neg \f\lor \p$.
 On the other hand, since $\f\vdash_k\p$ , again by the rule ($\lor \mathcal{I}1$), we get $\f\vdash_{k}\neg\f \lor \p$. This shows, by Definition \ref{kdepthlogic} that $\vdash_{k+1}\neg \f \lor \p$. 
Applying (c), we obtain that for any $\{\neg\f\lor \p\}$-maximal tree $T_{k+1}$ rooted in $\c$, we have $\a\vdash_0\neg\f\lor\p$ for each $\a \in Le(T_{k+1})$. 
 Now, let $\a \in Le(T_{k+1})$ be such that $\a\vdash_0\f$. By the rule $\neg\neg\mathcal{I}$ applied to $\f$, we obtain $\a\vdash_0\neg\neg \f$. Recalling that $\a\vdash_0\neg\f\lor\p$, applying the rule $\lor\mathcal{E}1$ to $\neg\f\lor\p$ and $\neg\neg\f$, we finally obtain $\a\vdash_0 \p$. }
 
  \end{proof}
\section{Approximating Probability Functions}\label{sec:approximating-probability-functions}

We now turn to our main concern, namely showing that finitely additive probability functions  $P\colon \fm \to [0,1]$  can be approximated by sequences of uncertainty measures which we will call \emph{k-depth belief functions}.

Depth-bounded trees capture the hypothetical uncertainty resolution of \emph{one} particular piece of information. However  the probabilistic quantification of uncertainty typically involve the simultaneous hypothetical exploration of more pieces of information.

Suppose, for instance, that an agent is informed that half of the balls in a urn are either yellow or green $(Y \lor G)$, and the remaining half are either red or white ($R\lor W$). Here, if an agent is willing to assign probabilities to the propositional variables $Y,G,R,W$, she needs to simulate information that she does not possess, concerning the piece of information $Y\lor G$ on the one hand, and $R\lor W$, on the other. For instance, on the basis of the symmetry of the problem, she might assign equal weights to $R$ and $W$, and to $Y$ and $G$. This would result in a probabilistic assignment $P$ such that $P(Y)=P(G)=P(R)=P(W)=1/4$. Once more, let us stress that the assignment $P(Y)=P(G)=P(R)=P(W)=1/4$ is different in nature from the information $P(Y\lor G) = P(R\lor W)= 1/2$: the former was arrived at by the agent, by taking into account different hypothetical scenarios and weighting them, while the latter was information actually provided to her from the start. This setup will be formally addressed, using our framework, in Example \ref{ellsbergvariant}.
  
The information actually held by an agent, as  that on $Y\lor G$ and $R\lor W$, in the above example, will be collected in the \emph{support set}, which we denote 
by $\init$. 
\paolo{By this we do not mean that an agent possesses information on whether the sentences in $\init$ are the case, but rather that only on these sentences the agent can  initially quantify her uncertainty probabilistically. 
}
 Each formula of $\init$ will act then as the root of a DBT-sequence, aimed at resolving the uncertainty therein contained.

With this motivation we now generalise the framework of Section \ref{sec:bounded} from sequences of depth-bounded trees to \emph{depth-bounded forests}. To do this, it will be convenient to assume that $\init$ is nonempty and contains at least one consistent sentence. As before, the symbol $\zinf$  will be used to denote the special case in which the agent holds no initial information at all, i.e. we let in that case $\init =\{\zinf\}$. In the following $inc(\Gamma)$ will denote the set $\{\a\in S\mid \a\vdash_0 \bot\}$ for any $\Gamma\subseteq \fm$. Unless otherwise stated, we shall thus have  $inc(\init)\subset \init \subseteq \fm \cup\{\zinf\}$.

\begin{definition}\label{Fk}
For each $\c \in \init$, let  $\{T_k(\c)\}_{k\in\mathbb{N}}$  be a DBT-sequence  rooted in $\c$. Then:
\begin{itemize}
\item A \emph{depth-bounded forest} sequence (DBF-sequence, for short) 
  ${(F_k)}_{k\in\mathbb{N}}$ based on $\init$ is obtained by letting
  $F_k = \cup_{\c\in\init} T_k(\c)$, for each $k\in\mathbb{N}$;

\item  A \emph{ finite depth-bounded forest sequence} is any initial
  segment ${(F_k)}_{k\leq n}$ of a depth-bounded sequence
  ${(F_k)}_{k\in\mathbb{N}}$, for some $n\in \mathbb{N}$.
\end{itemize}
\end{definition}

The notion of maximality and closure (see Definition \ref{max} above) adapts immediately to forests as follows.
\begin{definition}\label{maxf}
For $\f\in\fm, \spc\subseteq \fm $ we say that  
\begin{itemize}
\item a forest $F_k$ based on $\init$ is \emph{$\{\f\}$-closed} (maximal) if each tree $T_k(\gamma)$ in $F_k$ is $\{\f\}$-closed(maximal) for each $\c \in \init$.
\item a DBF-sequence $\{F_k\}_{k\in\mathbb{N}}$ based on $\init$ is \emph{$\spc$-closed} if each DBT-sequence $\{T_k(\c)\}_{k\in\mathbb{N}}$ in $\{F_k\}_{k\in\mathbb{N}}$ 
is $\spc$-closed, for each $\c\in\init$. 
\item \paolo{ a finite DBF-sequence $\{ F_k\}_{k\leq n}$ based on $\init$ is \emph{$\spc$-closed} (maximal) if each finite DBT-sequence $\{ T_k(\c)\}_{k\leq n}$ is $\spc$-closed (maximal), for each $\c\in\init$ .}
\end{itemize}
\end{definition}

We can now start working our way up to the approximation of probability functions. 
  \begin{definition}[Probability mass function]\label{m0}
For $\Gamma \subseteq \fm \cup \{\zinf\}$, say that  
$m\colon \mathcal{P}(\Gamma)  \to [0,1]$ is a \emph{probability mass function over $\Gamma$} if $m$ is a finitely additive measure, such that $m(\Gamma)=1$, $m(\emptyset) =0$, and $m(inc(\Gamma))=0$. 
\end{definition}

Note that in the interest of readability we write $m(\a)$ in place of $m(\{\a\})$, unless confusion may arise. We shall also refer to probability mass functions simply as ``mass functions''. 
\paolo{When applied to the set $\init$, the mass function defined above captures, in analogy with Dempster-Shafer Belief functions, the uncontroversial probabilistic quantification of uncertainty which is licensed by the actual information possessed by the agent.} While in Shafer's theory \cite{Shafer1990} agents can fill in potential gaps in their ``evidence'' by using judgment, the key idea of our proposal is to provide a logical framework to constrain this judgment formation. In particular, depth-bounded forests provide a principled way to allow agents to go beyond the information they actually possess, by exploring hypotheses for the resolution of uncertainty. The characteristic feature of our construction is that this simulation will be bounded by a \emph{fixed} sequence. 

\begin{definition}[DBM-sequence]\label{quantsequence} 
Let $\{F_k\}_{k\in\mathbb{N}}$ be a DBF-sequence based on $\init$, which is free of deep contradictions. We say that $\F = (F_k,m_k)_{k\in\mathbb{N}} $ is a \emph{depth-bounded
 mass sequence} (DBM-sequence for short) if
 each $m_k$ is a mass function over $Le(F_k)$, and for each $k >0 $, and leaf $\a \in Le(F_{k-1})$, the following holds:
 \begin{itemize}
 \item  $m_k(\a) =m_{k-1}(\a)$, if $\a\in Le(F_{k})$.
 \item $m_k(\{\a\land\b,\a\land\neg\b \}) =  m_{k-1}(\a)$, if $\a$ has children nodes  $\a\land\b$, $\a\land\neg\b$.
\end{itemize} 
\end{definition}

The condition on deep contradictions excludes the possibility that sentences might result consistent at depth $k-1$, while turning out inconsistent at depth $k$ (see \cite{Baldi2020,DAgostino2013}). If this were the case, the inconsistent sentences would have to be assigned mass 0, and their original mass should be redistributed among the consistent formulas. We do not consider that more general case, since it would complicate our definitions, without any major conceptual gains for our main results.

Note that the initial mass function $m_0$ is defined over $Le(F_0)$, which is equal to the initial support set $\init$. \paolo{The agent further refines the quantification by distributing the mass assigned to each leaf at depth $k$, among its children nodes, which are the leaves at depth $m_{k+1}$}. The leaves of depth-bounded forests for $k>0$ keep track of the information that has been probabilistically quantified by the agent. 

To extend this to an arbitrary sentence $\f$, the agent looks at the set of sentences whose uncertainty has been assessed (the leaves) and that allows her to 0-derive $\f$. Dually, the agent can also consider the set of sentences whose uncertainty has been assessed (the leaves) and that do not allow a 0-derivation of $\neg\f$. This is again reminiscent of the construction of Dempster-Shafer Belief functions.

Indeed, for any sentence $\f \in \fm$, given a DBF-sequence $\{F_k\}_{k\in \mathbb{N}}$ of depth-bounded forests, the set
\[b_{k}(\f) = \{\a \in  Le(F_k)  \mid \a \vdash_0 \f,  \
\a\not\vdash_0\bot\} 
\]
captures the belief mass that an agent commits to a sentence $\f$ \emph{after} simulating uncertainty resolution of depth $k$ whereas
 \[ pl_{k}(\f) = \{\a \in  Le(F_k)  \mid \a \not\vdash_0 \neg\f \}\]
similarly captures $k$-depth plausibility. 
This is the sense in which our construction puts the judgment component of Belief functions under logical constraints.

\begin{definition}\label{bk} 
 Let $(F_k,m_k)_{k\in\mathbb{N}}$ be a  DBM-sequence. The \emph{$k$-depth Belief function $B_k$} and the
  \emph{$k$-depth Plausibility function $Pl_k$} are defined by letting:
$$ B_k(\f) = m_k(b_k(\f)) \text{ and  }  Pl_k(\f) =  m_k(pl_k(\f)),$$ respectively.

\end{definition}
Let us now put our framework to work through a classic example. Though anticipated by \cite{Keynes1921,Knight1921} the example has gained prominence in the analysis of \cite{Ellsberg1961}, who used it to argue normatively against the adequacy of the probabilistic representation of uncertainty in the face of partly resolving uncertainty, or ambiguity. In addition to illustrating the basic mechanism of our framework, the example also sheds  new light on how the logical framework may contribute to the analysis of foundational issues in decision theory.

\begin{example}[Ellsberg's ``one urn'' problem]\label{ellsberg}
Let  $\L$ be a finite language with the propositional variables $\{Y,G,R\}$, interpreted as {\em the ball is Yellow, Green, Red}, respectively. The problem is that of a single draw from an urn and the ``Laplacian information'' is to the effect that 2/3 of the balls are either yellow or green with the remaining 1/3 being red. This information is logically represented by the conjunction $\c$ of the sentences in the set \[\{  Y\to (\neg G \land \neg R),  R\to (\neg G \land \neg Y), G \to (\neg R \land \neg Y), Y\lor G\lor R\}.\]

The  problem clearly suggests letting
$ \init = \{(Y\lor G) \land \c, R\land \c \}$.

The first step in our model now requires us to quantify our uncertainty based only on actual information, which is readily done by letting $m_0((Y\lor G)\land \c) = 2/3$ and $ m_0(R \land \c) =1/3$.  
This implies 
\begin{eqnarray*}
B_0(Y\lor G) =& 2/3\\
B_0(R)=&1/3\\
 B_0(\c) =& 1\\
B_0(Y)=B_0(G)=&0.  
\end{eqnarray*}

Since the interesting part of this problem is the process of analyzing the initial information and assigning probabilities to its sentences, it is natural to restrict the attention to sequences which are $S(\init)$-maximal. One can easily check that, for the node $R\land \c$ in $\init$, we already have $R\land \c\vdash_0\neg Y$,  $ R\land \c\vdash_0\neg G$ and  $ R\land \c\vdash_0 R$. The node thus decides  $S(\init)$, and  needs not be expanded at depth 1.  
However,  at depth 1 the node $(Y\lor R)\land\c$ will be suitably expanded. Note that there are only two possible $S({\init})$-maximal forests at this depth:
\begin{center}
\Tree[.$(Y\lor G)\land\c$ [.$Y$ ] [.$\neg Y$ ]] \ \ 
\Tree[.$R\land\c$ ] 
\end{center}
and 
\begin{center}
\Tree[.$(Y\lor G)\land \c$ [.$G$ ] [.$\neg G$ ]] \ \ 
\Tree[.$R\land\c$ ] 
\end{center}
In both cases all the leaves of the forest at depth 1 decide $S(\init)$, hence the forest is $S({\init})$-closed, and in particular maximal. \paolo{ At this depth, the symmetry of the available,  mutually exclusive, options may lead the agent to assign
$ m_1((Y\lor G)\land\c \land Y)) =m_1((Y\lor G)\land \c\land \neg Y)) = 1/3 $, though this is not the unique available option)}. This in turn yields
\begin{eqnarray*}
  B_1(Y)=&1/3\\
  B_1(G)=& 1/3\\
  B_1(R)=&1/3.
\end{eqnarray*}
Note that this is now compatible with the probabilistic quantification of uncertainty, which can be seen to have been obtained from information actually held by the agent, via a single  iteration of hypothetical uncertainty resolution.

\end{example}

\begin{example}\label{ellsbergvariant}
 The  variant of Ellsberg's problem sketched at the beginning of this Section arises by considering  the propositional variables $\{Y,G,R,W\}$, a sentence $\c$ representing background information, and taking
 $ \init = \{(Y\lor G) \land \c, (W\lor R)\land \c \}$.
 
 The quantification of uncertainty based on the information actually possessed by the agent is  rendered by letting
 $$m_0((Y\lor G)\land \c) = 1/2$$ and $$ m_0( (W\lor R) \land \c) =1/2.$$ 
This implies 
\begin{eqnarray*}
B_0(Y\lor G) =& 1/2\\
B_0(W\lor R)=&1/2\\
 B_0(\c) =& 1\\
B_0(Y)=B_0(G)=B_0(W)=B_0(R)=&0.  
\end{eqnarray*}

At depth 1 the nodes $(Y\lor G)\land\c$ and $(W\lor R)\land\c$ will be suitably expanded, for instance as follows:
\begin{center}
\Tree[.$(Y\lor G)\land\c$ [.$Y$ ] [.$\neg Y$ ]] \ \ 
\Tree[.$(W\lor R)\land\c$ [.$W$ ] [.$\neg W$ ]] \ \ 
\end{center}
At this point, at depth 1 all the leaves of the forest decide $S(\init)$. Again, by the symmetry of the problem, we might consider 
$ m_1((Y\lor G)\land\c \land Y)) =m_1((Y\lor G)\land \c\land \neg Y)) =  1/4 $ and
$ m_1((W\lor R)\land\c \land W)) =m_1((W\lor R)\land \c\land \neg W)) =  1/4 $ which in turn yields:
\begin{eqnarray*}
  B_1(Y)=&1/4\\
  B_1(G)=& 1/4\\
  B_1(R)=&1/4\\
  B_1(W)=&1/4.
\end{eqnarray*}
\end{example}

According to the main result of this Section, and indeed of the paper, $k$-depth Belief functions are shown to approximate their additive counterparts. In preparation to it we will illustrate the close relation between $k$-depth Belief functions and $k$-depth deductions.  We begin by adapting to the present framework Proposition 1 of \cite{Baldi2020}. The main difference between the two formalisms is that the $k$-depth belief functions here are based on a \emph{fixed} choice of hypothetical information, encoded in the underlying DBF-sequence, while those in \cite{Baldi2020}  may take into account any such choices at once. This owes to the rather different research questions which motivate the two papers. However, the two formalisms {are} 
 reconciled by adding suitable maximality assumptions to the DBF-sequence defining our functions.

\begin{proposition}\label{bkprop}
Let $\F = (F_k,m_k)_{k\in\mathbb{N}}$ be a DBM-sequence based on $\init\subseteq \fm \cup\{\zinf\}$. The following hold:
\begin{enumerate}
\item   $\vdash_k \f$  implies $B_k(\f)=1$, for any $F_k$ which is $\{\f\}$-maximal and free of deep contradictions. 
\item $\vdash_k \neg \f$  implies $B_k(\f)=0$, for any $F_k$ that is  $\{\f\}$-maximal  and free of deep contradictions. 
\item  $\f \vdash_k \p$  implies 
\begin{itemize}
\item  $B_{k}(\f) \leq B_{k}(\p)$, for some $F_k$. 
\item  $B_{k+1}(\f)\leq B_{k+1}(\p)$, for any $F_k$ which is $\{\neg \f \lor \p\}$-maximal and free of deep contradictions.
\end{itemize} 
\item$\displaystyle B_k( \bigvee_{i=1}^n \f_i ) \geq  \sum_{\emptyset \neq S\subseteq\{1,\dots,n\}} (-1)^{|S|-1} B_k(\bigwedge_{i\in S} \f_i)$ for each $\f_1,\dots,\f_n\in \fm$. 
\end{enumerate}
\end{proposition}
\begin{proof}
(1) and (2) are immediate consequences of Proposition \ref{treeandlogic} (b). The first item of (3) is an immediate consequence of Proposition \ref{treeandlogic}(c), and the second of Proposition \ref{treeandlogic}(d)

(4). By a straightforward adaptation of Proposition 1 in \cite{Baldi2020}. 
\end{proof}

\begin{remark}\label{rem:maximality}
 Proposition \ref{bkprop} makes explicit that suitable maximality assumptions
  play a fundamental role in establishing the close connection between
  $k$-depth Belief functions of Definition \ref{bk} and depth-bounded
  logics. The assumption requires that the agent chooses the \lq\lq right" pieces of hypothetical information: i.e. those which decide the sentence of interest in as many cases (i.e. leaves of the forest) as possible. This condition makes our results normatively interpretable: once the limitation of the inferential ability of the agent are taken into account (the depth), we are still requiring that her belief follows constraints dictated by the corresponding ($k$-depth) logic.
\end{remark}

We will now show that, under suitable assumptions, DBM-sequences determine a probability, in the limit.
 We begin with a lemma adapted from \cite{Baldi2021}, where it was originally presented in a qualitative setting.


\begin{lemma}\label{threshold}
\paolo{Let $\F = (F_k,m_k)_{k\in\mathbb{N}}$ be a DBM-sequence which is free of deep contradictions. 
For any formula $\f$, if there is a $j\in\mathbb{N}$ such that $F_j$ is $\{\f\}$-closed, then  $B_k(\f)= B_j(\f)$ for any $k\geq j$.}
\end{lemma}
\begin{proof}
Let $F_j$ be $\{\f\}$-closed, i.e. assume that each $\a\in Le(F_j)$ decides $\f$. Let us focus the set of leaves $\a$ such that $\a\vdash_0\f$, i.e. the set $b_j(\a)$. Now, for any $k\geq j$, the leaves which are descendants of $\a$ will be of the form $\a\land\b$. Hence, by applying rule $\land\mathcal{E}$ in Table \ref{intel}, and using the transitvity of $\vdash_0$  we will have $\a\land\b\vdash_0 \f$, for each such $\a\land\b$. 
A similar argument goes through for those remaining $\a\in Le(F_j)$ such that $\a\vdash\neg \f $. 
We get thus that $b_k (\f) = \{\a\land \b \in Le(F_k)\mid \a\vdash_0\f \}$. 
Moreover, since $\F$ is free of deep contradictions, we have that
$m_j(\a) = m_k\{\a\land\b \mid \a\land \b \text{ descendant of } \a\in Le(F_j)\}$ (see Definition \ref{quantsequence}) for each $\a\in b_j(\f)$ .
Hence, we obtain $m_k(b_k(\f)) = m_j(b_j(\f))$, and thus 
$B_k(\f) = B_j(\f)$ for each $k\geq j$.

\end{proof}

{We can now state and prove our main results. First, we show that
\paolo{ (additive) probability functions can be approximated by $k$-depth Belief functions. Conversely we identify the conditions under which $k$-depth Belief functions give rise to probability functions in the limit.}

 \begin{theorem}\label{thm:approxprob}
Let $P\colon \fm \to [0,1]$ be a probability function. Then there is a DBM-sequence $\F$  based on $\init = \{ \zinf\}$ such that $$P(\f) = \displaystyle\lim_{k\to\infty} B_k(\f),$$
for any $\varphi\in\fm$.
\end{theorem}
Something indeed quite stronger will be proved, namely that there is an $n\in\mathbb{N}$ such that $P(\f)= B_n(\f)$. Recall that we denote by $\at_\lang$ the set of  maximally (classical) consistent conjunctions of literals from a language $\mathcal L$ with $n$ propositional variables. 
\begin{proof}
Define a DBM-sequence $\mathcal{F} =
(F_k,m_k)_{k\in\mathbb{N}}$ based on $\zinf$ such that $Le(F_n)$ coincides with $\at_\lang$ and  $m_n(\a) = P(\a)$ for each $\a \in \at_\lang = Le(F_n)$. Note that once we fix $m_n$, Definition \ref{quantsequence} forces us to uniquely determine all the $m_k$ for $k< n$. Moreover, it can be easily shown that, if $\a\in At_\lang$,  then for any $\f\in \fm$, we have that $\a\vdash\f$ iff $\a\vdash_0\f$. 
Hence for any $\f\in \fm$ \[P(\f) = \sum_{\substack{ 
\a  \in  At_\lang \\ \a \vdash \f }} P(\a) = \sum_{\substack{ 
\a  \in  At_\lang \\ \a \vdash_0 \f }} P(\a) = m_n(b_n(\f))  = B_n(\f).
\] 
Moreover, at  depth $n$, the uncertainty about all the propositional variables in $\lang$ will have been resolved and $F_n$ will be $\{\f\}$-closed. 
Hence, by Lemma \ref{threshold} we will have $B_k(\f)=B_n(\f) = P(\f)$ for any $k\geq n$, as required.\end{proof}
\begin{remark}
Note that the sequence constructed in the proof above is in particular an $\fm$-closed sequence, which at a certain depth decides all the propositional variables, hence all the sentences, in our finite language.
The argument extends also to a countable language. To see this, note that by Lemma \ref{threshold}, for each
  sentence $\f \in \fm$ there is a $j_\f\in\mathbb{N}$ such that   $P(\f) = B_{j_\f}(\f)$, and $B_k(\f) = B_{j_\f}(\f)$ for
  each $k\geq j_\f$.  Hence, what is peculiar to the countable
  case is the fact that the index $n$ at which $B_n(\f)$s equals
  $P(\f)$ may be distinct for distinct elements of $\fm$. Whether $\L$
  is finite or countable,
  $P(\f) = \displaystyle\lim_{k\to\infty} B_k(\f)$.
\end{remark}

\begin{theorem}\label{thm:prob}
\paolo{Let $\F = (F_k,m_k)_{k\in\mathbb{N}}$ be a DBM-sequence with $\{F_k\}_{k\in\mathbb{N}}$ a $\fm$-closed DBF-sequence free of deep contradictions. 
 Then, the function $P\colon\fm\to[0,1]$,  obtained by letting \[ P(\f) = \displaystyle\lim_{k\to\infty} B_k(\f)\]  
for any $\varphi\in\fm$, is a probability.}
\end{theorem}
\begin{proof}
To prove that $P$ is a probability function, we need to show that it satisfies the two properties (PL1) and (PL2) in the introduction.

For (PL1), assume $\vdash\f$ for some $\f\in\fm$. Then, since the DBLLs approximate $\vdash$ \cite{DAgostino2013} there is a $i$ such that $\vdash_i\f$. 
On the other hand, since $\{F_k\}_{k\in\mathbb{N}}$ is $\fm$-closed, there is a  $\{\f\}$-closed $j$-depth forest $F_j$ free of deep contradictions. Hence, by Proposition \ref{bkprop} we will then have  $B_j(\f) =1$, and by Lemma \ref{threshold}, $B_n(\f) =1$ for each $n\geq \max\{j,i\}$  . 

Let us now assume that $\vdash\neg(\f\land\p)$. We want to show that $P(\f\lor\p)=P(\f)+P(\p)$. 
As before, for some $i$, we will have that $B_i(\f\land \p ) =0$, and the same will hold   for any $n\geq i$. 

Now, since $\F$ is $\fm$-closed, there will be a $j$ such $F_j$ is $\{\f\}$-closed. 
and a $k$ such that $F_k$ is $\{\p\}$-closed. Now, let us consider any $n\geq \max\{i,j,k\}$. First, since $n\geq i$, we have $B_n(\f\land\p)=m_n(b_n(\f\land\p)) = 0$. Let us now consider the set $b_n(\f\lor\p)$. 
Since $n\geq \max\{j,k\}$, $F_n$ will be both $\{\f\}$-closed and $\{\p\}$-closed, hence we will have 
%

\begin{align}\label{bnsum}
\begin{split}
b_n(\f\lor\p) & =  \{ \a\in Le(F_n) \mid\a\vdash_0 \f\lor\p)\} \\ & =  \{ \a\in Le(F_n) \mid\a\vdash_0 \f, \a \vdash_0 \p)\} \\  &  \cup \{\a\in Le(F_n)\mid \a\vdash_0\f, \a\vdash_0\neg\p\} \\  
& \cup \{\a\in Le(F_n)\mid \a\vdash_0\neg \f, \a\vdash_0\p\} 
\end{split} 
\end{align}
where all three sets are clearly disjoint. On the other hand, we have 
\[ \{\a\in Le(F_n)\mid \a\vdash_0\f, \a\vdash_0\p\} = \{ \a\in Le(F_j)\mid \a\vdash_0\f\land\p\} \]
and  $m_n( \{\a\in Le(F_n)\mid \a\vdash_0\f\land\p\})= m_n(b_n(\f\land\p))= 0$. 
From this and Equation \ref{bnsum},  it follows that $m_n(b_n(\f\lor\p))$ equals:
\begin{align*}
  m_n \{\a\in Le(F_n)\mid \a\vdash_0\f, \a\vdash_0\neg\p\} + m_n
 (\{\a\in Le(F_n)\mid \a\vdash_0\neg \f, \a\vdash_0\p\})\\ = m_n(b_n(\f)) + m_n(b_n(\p))
\end{align*} 
From the latter, we finally get $B_n(\f\lor\p) = B_n(\f) + B_n(\p)$. This will hold for any $n\geq \max \{i,j,k\}$, hence  we finally obtain our desired result, i.e. $P(\f\lor\p) = P(\f) + P(\p)$.
 \end{proof}

\section{When $k-$depth Belief functions are tractable}\label{sec:complex}
We now turn to identifying the conditions under which the approximations granted by Theorem \ref{thm:approxprob} are computationally feasible. 
\begin{definition}\label{uniform}
\paolo{A sequence $\F = (F_k)_{k\in\mathbb{N}}$ of depth-bounded forests with initial support $\init$ is said to be:
\begin{itemize}
\item \emph{uniform} iff, for each $k$, all the nodes in the trees in $F_k$ contain the same hypothetical information, with the exception of the root. 
\item \emph{analytic} iff it is $S(\init)$-maximal, and all the nodes in the trees in $F_k$ contain as hypothetical information only formulas in $S(\init)$. 
} 
\item A \emph{$\spc$-maximal uniform(analytic) sequence} is a sequence which is $\spc$-maximal among the uniform(analytic) sequences.
\end{itemize}
\end{definition}
Note that the restriction to uniform sequences is immaterial for the result in Theorem \ref{thm:approxprob}, since the approximating sequence built in the proof is based on $\{\zinf\}$, hence it is trivially uniform. 

\begin{example}
Let $\init = \{p\lor q, s \lor r\}$. The 1-depth forest
\begin{center}
\Tree[.$(p\lor q)$ [.$p$ ] [.$\neg p$ ]] \ \ 
\Tree[.$(s\lor r)$ [.$p$ ] [.$\neg p$ ]] \ \ 
\end{center}
is uniform, while the 1-depth forest
\begin{center}
\Tree[.$(p\lor q)$ [.$p$ ] [.$\neg p$ ]] \ \ 
\Tree[.$(s\lor r)$ [.$s$ ] [.$\neg s$ ]] \ \ 
\end{center}
is not. Moreover, the forests in Example \ref{ex:maxnomax} are all trivially uniform, since they reduce to just single trees, while the forests at Depth 1 in Example \ref{ellsberg} are not uniform. 
\end{example}
Let us now consider two core reasoning tasks with our depth-bounded belief functions.  
Following \cite{Paris1994,finger15,hansen2000,FHM90}, which investigate the complexity of reasoning with probability functions, 
we assume that an agent is provided with $n$ linear constraints on the sentences $\f_1,\dots,\f_p$, of the form:
\begin{equation}\label{preclsystem}
\sum\limits_{j=1}^p d_{ij} B(\f_j) \leq z_i \ \  \ \  i=1,\dots,n \ \ \ d_{ij}, z_i \in \mathbb{Q}.
\end{equation}
Let us equivalently rewrite \eqref{preclsystem} as
\begin{equation}\label{clsystem}
\sum\limits_{j=1}^{p_1} a_{ij} B(\f_j) + w_i \leq \sum\limits_{j=1}^{p_2} b_{ij} B(\f_j) +v_i
\end{equation}
where $i=1,\dots,n \ \ \ a_{ij},b_{ij},w_i,v_i \in \mathbb{N},  p_1+p_2 =p.$
The GENPSAT problem (see e.g. \cite{CCM17}) consists in finding out whether there is a probability function over $\fm$ satisfying such a system. It is well-known that GENPSAT, as well as other minor variants involving linear constraints, is NP-complete \cite{Paris1994,finger15}. 
For our present purposes, a variation on the problem will be more appropriate: \paolo{we will rather aim to find out whether there are depth-bounded Belief and Plausibility functions, which can be seen as approximations of a probability function satisfying the given linear constraints.} We will thus consider decision problems which stand to our $k$-depth logics and $k$-depth belief functions as the GENPSAT problem stands to classical logic and classical probability functions.

Let us start from the 0-depth case: \\

\noindent {$\mathrm{GENSAT_0}$ Problem } \\
\noindent{\bf INPUT}: The set of $p$ sentences and $n$ linear constraints in \eqref{clsystem}.\\
\noindent {\bf PROBLEM}: 
Find a 0-depth belief function over $\init =\{ \f_1,\dots,\f_p\}$ such that for each $i\in\{1,\dots,n\}$ :
$$ \sum\limits_{j=1}^{p_1} a_{ij} B_0(\f_j) + w_i \leq \sum\limits_{j=1}^{p_2} b_{ij} Pl_0(\f_j) +v_i$$ 
where $ a_{ij},b_{ij},w_i,v_i \in \mathbb{N}$.  

By Definition \ref{bk} the problem boils down to finding a solution for the following $n$ system of linear inequalities in unknowns $m_0(\f_1),\dots, m_0(\f_p)$:
\begin{equation}\label{0system}
\begin{aligned}
\sum\limits_{j=1}^{p_1} a_{ij} \sum\limits_{\substack{k= 1,\dots,p_1 \\ \f_k \vdash_0 \f_j}} m_0(\f_k) + w_i & \leq   \sum\limits_{j=1}^{p_2} b_{ij} \sum\limits_{\substack{k= 1,\dots,p_2 \\ \f_k \not\vdash_0 \neg\f_j}} m_0(\f_k) + v_i   &  i=1,\dots,n \\
m_0(\f_j)& =0 & \text{ \ if \ } \f_j\vdash_0\contr.
\end{aligned}
\end{equation}

Let us recall that $|inc(\init)|= |\{\a\in \init\mid \a\vdash_0 \bot\}|$ stands for the count of inconsistent sentences in the support.
\begin{theorem}\label{poly0}
The $\mathrm{GENSAT_0}$ problem is solvable in $\mathbf{PTIME}(|S(\init)|+n)$.
\end{theorem}
\begin{proof}
By results in \cite{DAgostino2013}, finding out whether $\f_k \vdash_0 \f_j$, $\f_k\not\vdash_0\neg\f_j$ and  whether $\f_j\vdash_0\contr$ requires time polynomial in $|S(\init)|$. Note that the system \eqref{0system} has size $p \times (n+p+1+|inc(\init)|)$, so finding a solution requires polynomial time as well. Hence our problem is in {\bf PTIME}($|S(\init)|+n$).
 \end{proof}
Let us now consider the problem of finding out whether there is a $k$-depth belief function, for a given $k>0$, which approximates a probability function satisfying the constraints in \eqref{clsystem}. The idea is to fix $k$ and an initial $\init$, and then consider all the $S(\init)$-maximal depth-bounded forests $F_k$, in order to check whether there is a depth-bounded Belief function $B_k$ defined on such forests, which satisfies the constraints. We will take as initial support the set $\init = \{\f_1,\dots,\f_p\}$ and assume that the forests are uniform. Note that this assumption prevents the number of forests to grow exponentially with $|\init|$.
Let us now fix a $k>0$ and consider the following problem:\newline

\noindent{ $\mathrm{GENSAT_k}$ Problem.}

\noindent{\bf INPUT}: The $n$ constraints in \eqref{clsystem}.

\noindent {\bf PROBLEM}: Find out whether there is a finite
DBM-sequence $\F = (F_j,m_j)_{j\leq k}$, where $\{F_j\}_{j\leq k}$ is a uniform, analytic,
 finite DBF-sequence based on $\init = \{ \f_1,\dots,\f_p \}$, such that 
$$ \sum\limits_{j=1}^{p_1} a_{ij} B_k(\f_j) + w_i \leq \sum\limits_{j=1}^{p_2} b_{ij} Pl_k(\f_j) +v_i$$
where $i=1,\dots,n \ \ \ a_{ij},b_{ij},w_i,v_i \in \mathbb{N},  p_1+p_2 =p.$

Note that, also in this case with $k>0$, we have still taken $\init = \{\f_1,\dots,\f_p\}$. Since these are the sentences involved in the formulation of the constraints, we can take them to be the initial information provided to an agent, i.e. the information that can be probabilistically quantified by the agent without the use of further hypothetical information.

Recalling Definition \ref{bk}, solving the problem $GENSAT_k$ amounts to 
finding an $F_k$ based on $\init = \{\f_1,\dots,\f_p\}$, such that the following system, in the unknowns $m_k(\a)$, for $\a \in Le(F_k)$, has a solution: 
\begin{equation}\label{ksystem}
\begin{aligned}
\sum\limits_{j=1}^{p_1} a_{ij} \sum_{\substack{\a\in b_k(\f_j) }} m_k(\a) + w_i & \leq \sum\limits_{j=1}^{p_2} b_{ij} \sum_{\substack{\a\in pl_k(\f_j) }} m_k(\a) + v_i & \text{ for each } i=1,\dots,n\\
m_k(\a) & \geq 0 & \text{ for each } \a \in Le(F_k) \\
\sum\limits_{\a\in Le(F_k)} m_k(\a) & = 1 & \\
m_k(\a)& =0 & \text{ if } \a\vdash_0\contr
\end{aligned}
\end{equation}
\begin{theorem}\label{polyk}
$\mathrm{GENSAT_k}$ is solvable in $\mathbf{PTIME}(|S(\init)|+n) $.
\end{theorem}
\begin{proof} 


Consider the number of possible forests of depth $k$, based on $\init$. Since we assume that the forests are uniform, this reduces to the number of possible trees, which in turn reduces to the possible choices of formulas for expanding each node. The number of such choices at depth $k$ will be $2^k-1$ and, since we restrict to analytic sequences, we will only consider sentences in $S(\init)$ as candidates. We thus obtain that the number of possible forests at depth $k$ is bounded above by $|S(\init)|^{2^k-1}$.

We then select among those $|S(\init)|^{2^k-1}$ different forests, those which are maximal with respect to  $S(\init)$, see Definition \ref{max}. This is verified  by running, for each forest, and for each formula $\f$ in $S(\init)$, the quadratic algorithm  e.g. in \cite{DAgostino2013}, to verify whether the leaves decide $\f$. 

For each remaining selected forest $F_k$ of depth $k$, the number of leaves in $Le(F_k)$ will be bounded above by $|\init|\cdot 2^k$, which is linear in $|\init|$ ( once $k$ is fixed, $2^k$ is constant).

We then need to check whether the set of sentences in $Le(F_k)$ may satisfy system \eqref{ksystem}. Each such system has size bounded by $2^k \times (n+ 2^k +1+ |inc(Le(F_k))|)$. 
Hence we have polynomially many systems of polynomial size. Finally, since solving each such system still requires polynomial time (by standard linear algebraic methods), we obtain our claim.
\end{proof}

The problems considered so far are decision problems, and in particular variants of the well-known SAT problem and GENPSAT problem for propositional satisfiability, and (generalized) probabilistic satisfiability, respectively. Another class of problems is also quite relevant in the area of probabilistic logic, as presented in \cite{Hailperin1996} and in subsequent developments \cite{Haenni2011}. The idea is to provide a probabilistic generalization of inference, by addressing the following problem:
\begin{center}
Let $\f_1,\dots, \f_p,\f$ be sentences. Take any probability function $P$ such that $P(\f_1),\dots, P(\f_p)$ obey a given set of constraints. What kind of constraints follow on $P(\f)$?
\end{center} 
In \cite{Hailperin1996,Haenni2011} 
 the constraints imposed on the probabilities of $\f_1,\dots,\f_p$ are typically  intervals, and the goal is to identify the ``smallest'' interval $[a,b]\subseteq [0,1]$ such that $P(\f)\in[a,b]$, whenever $P$ satisfies the constraints on the premises. This reduces in the classical setting to a linear programming problem. We will consider here a more general problem, allowing constraints as in \eqref{clsystem}, and searching for solutions based on depth-bounded belief functions. We have then the following: \\

\noindent{\bf $B_k$INF Problem}\\

\noindent{\bf INPUT}: The $n$ constraints in \eqref{clsystem}, and an arbitrary sentence $\f$.

\noindent {\bf PROBLEM}: Compute the smallest interval $[a,b]$ such that, for any
 $B_k, Pl_k$  which satisfy the constraints in \eqref{clsystem} as in the $\text{ GENSAT}_k$ Problem, 
 $B_k,Pl_k$ satisfy $a\leq B_k(\f)$, $Pl_k(\f) \leq b$.

As we now show, we can just adapt the classical techniques based on linear programming, and Theorem \ref{polyk} in this section, to obtain that also the functional problem $B_k\text{INF}$ is solvable in polynomial time. Again, we will let $\init =\{\f_1,\dots,\f_m\}$,
\begin{theorem}\label{inf}
The problem $B_k\text{INF}$ is solvable in $\mathbf{FPTIME}(|S(\init)|+ S(\f)+ n)$.
\end{theorem}
\begin{proof}
By the argument in the proof of Theorem \ref{polyk}, we identify polynomially many analytic uniform DBF-sequences based on $\init = \{\f_1,\dots,\f_m\}$. For each such system our problem is solved by considering the following linear programming problem. 
\begin{equation}\label{fkminimize}
\begin{aligned}
\textbf{ minimize  }  &     \sum_{\a\in b_k(\f) } m_k(\a)
  \\
\text{ subject to } & \\
\sum\limits_{j=1}^{p_1} a_{ij} \sum_{{\a\in b_k(\f_j) }} m_k(\a) + w_i & \leq \sum\limits_{j=1}^{p_2} b_{ij} \sum_{{\a\in pl_k(\f_j) }} m_k(\a) + v_i & \text{ for each } i=1,\dots,n\\
m_k(\a) & \geq 0 & \text{ for each } \a \in Le(F_k) \\
\sum\limits_{\a\in Le(F_k)} m_k(\a) & = 1 & \\
m_k(\a)& =0 & \text{ if } \a\vdash_0\contr
\end{aligned}
\end{equation}

and similarly

\begin{equation}\label{fkmaximize}
\begin{aligned}
\textbf{ maximize  }  &     \sum_{\a\in pl_k(\f) } m_k(\a)
  \\
\text{ subject to } & \\
\sum\limits_{j=1}^{p_1} a_{ij} \sum_{{\a\in b_k(\f_j) }} m_k(\a) + w_i & \leq \sum\limits_{j=1}^{p_2} b_{ij} \sum_{{\a\in pl_k(\f_j) }} m_k(\a) + v_i & \text{ for each } i=1,\dots,n\\
m_k(\a) & \geq 0 & \text{ for each } \a \in Le(F_k) \\
\sum\limits_{\a\in Le(F_k)} m_k(\a) & = 1 & \\
m_k(\a)& =0 & \text{ if } \a\vdash_0\contr
\end{aligned}
\end{equation}
Solving the two problems will provide our desired bounds for $\f$. This can be achieved in polynomial time invoking algorithms for linear programming \cite{Schrijver1987}, since the size of the systems to be solved is polynomial as well in the size of $\f_1,\dots,\f_p,\f$.
\end{proof}
\section{Conclusions}
\label{sec:conclusion}

The results of this paper suggest that there should be no methodological tension between requiring that
\begin{itemize}
\item[i)] models of uncertain reasoning should be aimed at
  realistic, rather than idealised agents;
\item[ii)] the quantification of uncertainty should use all the information actually possessed by an agent and
\item[iii)] the quantification of uncertainty should keep track of the reasoning which leads to the hypothetical resolution of uncertainty -- i.e. using information which the agent does not actually possess but which can be obtained by suitably exploiting reasoning by cases.
\end{itemize}

In consonance with the general view  outlined in \cite{DAgostino2017}  the results illustrated in this paper suggest  that a key to achieving such a  realistic model of uncertain reasoning lies in grounding it in logics which are sensitive to information, rather than classical truth-conditions.

\hykel{Not only our main results may lead to a reconciliation between additive and non-additive approaches to uncertain reasoning, but they also suggest that logic may play a much more prominent role in uncertain reasoning compared to what is seen in current practice.  Whilst AI has long benefitted methodologically from logical approaches, Economic theory has done this to a much lesser extent. Building on Examples \ref{ellsberg} and \ref{ellsbergvariant} future work will focus on this promising application of logical methods in uncertain reasoning.} 
%

As to AI applications, future work will be devoted 
first to the problem of determining, for a specific application-driven problem, the bounds $k$ for which the uncertainty quantification provided by $B_k$ turns out to be adequate. In addition, better algorithms for the problems in Section \ref{sec:complex} are desirable and will be investigated. Let us note that the proofs in Section \ref{sec:complex} contain essentially brute-force search algorithms, where the search space is bounded by the fixed required depth, and only this prevents an exponential explosion. This leaves open the investigation of algorithms which are more tightly calibrated to exploring only the necessary maximal forests. 
In future work, we will develop such algorithms and provide suitable time estimations


\paolo{Finally, as we pointed out in Remark \ref{rem:maximality} the notion of maximal hypothetical uncertainty resolution makes our framework interpretable normatively.} Moreover the tractability results of Section \ref{sec:complex} suggest that reasoning with depth-bounded belief functions may be within the means of realistic agents. Hence, the framework introduced in this paper may be taken as a benchmark for the empirical analysis of what experimental subjects should be expected to do when performing uncertain reasoning tasks. 




\section*{Acknowledgments}
This research was funded by the Department of Philosophy \lq\lq Piero Martinetti" of the University of Milan under the Project \lq\lq Departments of Excellence 2018-2022" awarded by the Ministry of Education, University and Research (MIUR). We are grateful to Marcello D'Agostino for very useful conversations on the theory of Depth-Bounded Boolean Logics which motivated a substantial part of this research. 
We are also very grateful to two anonymous reviewers for their unrestrained criticisms, which certainly lead this to becoming a much better paper.

\end{document}